\documentclass[reqno,a4paper]{amsart}
\renewcommand{\paragraph}{\subsubsection*}
\usepackage[UKenglish]{babel}
\usepackage{microtype}
\usepackage{amssymb}
\usepackage{amsmath}
\usepackage{amsthm}
\usepackage{amsfonts}
\usepackage{xcolor}
\usepackage{booktabs}

\usepackage[foot]{amsaddr}
\usepackage{lineno}
\usepackage{url}
\usepackage[
  pdftex,
  bookmarks=true,
  bookmarksnumbered=true,
  hypertexnames=false,
  breaklinks=true,
  linkbordercolor={0 0 1}
  ]{hyperref}

\providecommand{\urlstyle}[1]{}
\providecommand{\doi}[1]{\href{http://dx.doi.org/#1}{\nolinkurl{doi:#1}}}

\usepackage{pgf}
\usepackage{pgfplots}
\usepackage{tikz}
\usepackage{todonotes}
\usetikzlibrary{arrows,calc,automata,backgrounds,decorations,intersections,fillbetween,shapes.misc,patterns, decorations.pathreplacing, calc}

\usepackage{cleveref}
\usepackage{caption}
\usepackage{subcaption}
\usepackage{thm-restate}
\usepackage[makeroom]{cancel}

\newtheorem{definition}{Definition}
\newtheorem{example}[definition]{Example}
\newtheorem{theorem}{Theorem}

\newtheorem{proposition}[theorem]{Proposition}

\newtheorem{conjecture}[theorem]{Conjecture}

\newcommand{\NN}{\mathbb{N}}
\newcommand{\ZZ}{\mathbb{Z}}
\newcommand{\QQ}{\mathbb{Q}}
\newcommand{\RR}{\mathbb{R}}
\newcommand{\QQbar}{\overline{\mathbb{Q}}}
\newcommand{\CC}{\mathbb{C}}

\newcommand{\rem}{\mathrm{rem}}

\newcommand{\As}{A}
\newcommand{\Bs}{B}
\newcommand{\Supp}{\mathrm{Supp}}

\begin{document}

\title[MP for Hypergeometric Sequences with Rational Parameters]{The Membership Problem for Hypergeometric Sequences with Rational Parameters}
\author{Klara Nosan,$^1$ Amaury Pouly,$^{1,2}$, \mbox{Mahsa
  Shirmohammadi,$^1$} \and \mbox{James Worrell$^2$}}
\address{\begin{minipage}{\textwidth}
$^1$~Universit\'e Paris Cité, CNRS, IRIF, F-75013 Paris, France\\
$^2$~Department of Computer Science, University of Oxford, UK\end{minipage}}
\date{}

%
%
%
%

\begin{abstract}
We investigate the Membership Problem for hypergeometric sequences: given a hypergeometric sequence 
$\langle u_n \rangle_{n=0}^\infty$ of rational numbers and a target $t \in \mathbb{Q}$, decide whether 
$t$ occurs in the sequence.  We show decidability of this problem under the assumption that in the defining recurrence $p(n)u_{n}=q(n)u_{n-1}$, the roots of the polynomials $p(x)$ and $q(x)$ are all rational numbers. Our proof relies on bounds on the density of primes in arithmetic progressions. 
We also observe a relationship between the decidability of the Membership problem (and variants) and the Rohrlich-Lang conjecture in transcendence theory. 
\end{abstract}


\maketitle

\DeclareRobustCommand{\gobblefive}[5]{}
\DeclareRobustCommand{\gobblenine}[9]{}
\newcommand*{\SkipTocEntry}{\addtocontents{toc}{\gobblenine}}

\section{Introduction}

Recursively defined sequences are natural objects of study in computation,
arising in the analysis of algorithms, weighted automata, loop termination, and probabilistic models, among many other areas.
In this context, the following decision problem frequently arises: \emph{does a given value appear in a given recurrence sequence?} Perhaps, the most famous example of this type of membership (or reachability) problem is the Skolem Problem, which asks to decide the existence of a zero term in a given linear recurrence sequence (with constant coefficients).
For the majority of problems of this kind, their decidability status is wide open and apparently very challenging.  For example, the Skolem Problem has been regarded as open since at least the 1970s, with decidability only known 
for linear recurrences of order at most four~\cite{MST84,Ver85}.

In this paper we are concerned with the Membership Problem for hypergeometric sequences, which are those satisfying an order-1 recurrence 
$p(n)u_{n}=q(n)u_{n-1}$
with polynomial coefficients $p(x)$ and $q(x)$.  In other words, we focus on order-1 polynomially recursive sequences: arguably the simplest class of recurrence sequences which are not linearly recursive.
The Membership Problem here asks, given a hypergeometric sequence $\langle u_n \rangle_{n=0}^\infty$ and a target $t$, whether there exists $n$ with $u_n=t$.  

Decidability of this problem
may appear trivial at first glance. Indeed, 
the sequence $\langle u_n \rangle_{n=0}^\infty$ either diverges in absolute value or converges to a 
value in $\mathbb{R}$.  If the sequences does not converge to~$t$
then it is not difficult to compute a bound $N$ such that $u_n \neq t$ for all $n>N$.  But such a bound can also be computed in case one knows that $\langle u_n \rangle_{n=0}^\infty$ converges to $t$ (by straightforward arguments about the monotonicity of the convergence).  The difficulty with this analysis is that it depends on being able to distinguish the two cases above, i.e., that it be
semi-decidable 
whether $\langle u_n \rangle_{n=0}^\infty$ converges to~$t$.  We explore this question in Section~\ref{sec:rohrlich}, observing its connection to the Rohrlich-Lang conjecture in number theory (following ideas similar to~\cite[Section 4]{kenison2020positivity}).

In the rest of the paper
we approach the decidability of the membership problem from a different angle---specifically, by considering the prime divisors of $u_n$.  Our strategy is to show that (except in some degenerate cases) for all sufficiently large $n$, $u_n$ has a prime divisor $p$ that is not also a prime divisor of the target $t$.  This allows us to compute a bound $N$ such that $u_n\neq t$ for all $n>N$.  We obtain this bound using classical results on the distribution of primes in arithmetic progressions.  However our proof requires that the 
polynomial-coefficients $p(x)$ and $q(x)$ of the order-1 recurrence that specifies the sequence 
$\langle u_n \rangle_{n=0}^\infty$ have rational roots.  Of course, in general, these roots will be algebraic numbers. We remark that such rationality assumptions are a feature of work on divisibility properties of hypergeometric sequences (and their associated generating functions, which are hypergeometric series);
see, for example,~\cite{DRR13}.  

\subsection{Related work}
The paper~\cite{AMM07} studies the asymptotic behaviour of  
    sequences of the form $v_p(u_n)$, where $v_p$ is the $p$-adic valuation and 
    $\langle u_n\rangle_{n=0}^\infty$ is a sequence of integers satisfying a 
    recurrence $u_{n} = q(n)u_{n-1}$, i.e.,
    an order-1 polynomial recurrence whose leading coefficient is constant.
By contrast, we consider general order-1 recurrences, but consider only divisibility or non-divisibility by a well-chosen set of primes.

The problem of deciding positivity of order-2 polynomially recursive 
    sequences and of deciding the existence of zeros in such sequences is considered in~\cite{KauersP10,KenisonKLLMOW021,NeumannO021,PillweinS15}.  These works all place syntactic restrictions on the 
    degrees of the polynomial-coefficients involved in the recurrences, and all four give algorithms that are not guaranteed to terminate for all initial values of a given recurrence (essentially due to phenomena similar to that explored in Section~\ref{sec:rohrlich}).
 %
   
    In~\cite{KauersP10} two algorithms are described   which can be used for proving positivity of P-finite recurrences. The termination of these algorithms is proved for restricted classes of balanced recurrences of order up to three;
a recurrence is balanced if the leading and trailing coefficient have the same degree and  all other coefficients are bounded by this degree. 
The termination guarantee for order-2 balanced recurrences requires that the characteristic roots be distinct and that one is working with a
generic solution of the recurrence (in which rate of growth corresponds
to the dominant characteristic roots).        
    Given a converging hypergeometric sequence $\langle u_n\rangle_{n=0}^\infty$ and a target value $t$, the sequence $u_n-t$ satisfies a balanced second order recurrence, but  the sequence $u_n -t$  will
have a double characteristic $1$ in the "hard cases" -- namely when
$u_{n+1}/u_n$ converges to~$1$.

    A polynomially recursive sequence
    is said to be a \emph{closed form} 
    if it is the sum 
    of hypergeometric sequences~\cite[Definition 8.1.1]{PetkovsekWilfZeilberger1996}.
        Using the fact that
    the quotient of two hypergeometric sequences is again hypergeometric, one can easily deduce that the problem of deciding the existence of a zero in a closed-form order-2 polynomially recursive sequence reduces to the membership problem for hypergeometric sequences.
    
    The paper~\cite{BBY11} proves a version of the Skolem-Mahler-Lech paper for a subclass of polynomially recursive sequences.
        Specifically for a sequence $\langle u_n\rangle_{n=0}^\infty$
    satisfying a polynomial recurrence $u_n = \sum_{k=1}^d p_k(n) u_{n-k}$, under the assumption that $p_d$ is a non-zero constant polynomial, it is shown that $\{ n \in \mathbb{N} : u_n =0\}$ is the union of a finite set and finitely many arithmetic progressions.  It remains open whether this conclusion extends to the class of all polynomially recursive sequences.  The proof of this result in~\cite{BBY11} 
uses Strassman's Theorem in $p$-adic analysis and appears not to give information of how to decide the existence of a zero in a given sequence.

\section{The Membership Problem} \label{sec:mp}
We denote by $\QQ[x]$ the ring of univariate polynomials with rational coefficients, and by $\QQ(x)$ the field of univariate rational functions with rational coefficients.

An infinite sequence $\langle u_n \rangle_{n=0}^\infty$ of rational numbers is called a univariate \emph{hypergeometric sequence} if it satisfies a recurrence of the form 
 \begin{equation}\label{eq:rel}
 p(n)u_{n} - q(n)u_{n-1} = 0 \, ,
\end{equation} 
 where $p(x),$ $q(x) \in \QQ[x]$ are polynomials, and  $p(x)$
 has no non-negative integer zeros. By the latter assumption on~$p(x)$, the  recurrence relation~\eqref{eq:rel} 
 uniquely defines an infinite sequence of rational numbers once the initial value $u_0\in\QQ$ is specified.  We say that such an induced sequence is a \emph{hypergeometric sequence with rational parameters} if both $p(x)$ and $q(x)$ split completely over~$\QQ$.

Recurrence~\eqref{eq:rel} can be reformulated as follows:
$$u_{n} = r(n)u_{n-1} \, ,$$
where $r(x)\in \QQ(x)$ is  a rational function that, 
by the assumption above, has no nonnegative integer pole.  The rational function~$r(x)$ is called the \emph{shift quotient} of~$\langle u_n \rangle_{n=0}^\infty$.

 Let us introduce the decision problem we seek to investigate.  We say that $t\in \mathbb{Q}$ is a \emph{member} of 
 a sequence $\langle u_n \rangle_{n=0}^\infty$ if there exists $n\in \NN$ such that $u_n=t$; we further refer to $n$ as an index of~$t$ in the sequence.  The \emph{Membership Problem (MP)} for hypergeometric sequences is the problem of deciding, given a hypergeometric sequence $\langle u_n \rangle_{n=0}^\infty$ (specified by a a recurrence of the form~\eqref{eq:rel} with a given initial value $u_0$) and a target $t\in \mathbb{Q}$, whether $t$ is a member of~$\langle u_n \rangle_{n=0}^\infty$.  Our main contribution in this paper is to show that MP is decidable for hypergeometric sequences with rational parameters. Our solution in fact allows to compute the set of all indices of~$t$ in the  sequence~$\langle u_n \rangle_{n=0}^\infty$.

We will assume that in instances of MP, the
numerator~$q(x)$ of the shift quotient
has no non-negative integer zeros and that the target $t$ is non-zero.  This assumption is without loss of generality for the decidability results as discussed in Appendix~\ref{appendix:zero}.

\subsection{Limit of Shift Quotients}\label{subsec:shiftquo}

Given an instance of MP, consisting of a rational value~$t$ and a hypergeometric sequence $\langle u_n \rangle_{n=0}^\infty$ with rational parameters and initial value~$u_0$, our general approach to solving MP is to show that there is an effectively computable upper bound~$N$, depending only on the sequence $\langle u_n \rangle_{n=0}^\infty$, such that $u_n=t$ only if $n \leq N$.  This reduces  the membership test for $t$ to that of searching within the finite set $\{u_0,\ldots, u_N\}$ and hence entails the decidability of MP and computability of the set of indices of $t$ in $\langle u_n \rangle_{n=0}^\infty$.

To obtain the bound $N$, we first note that as $x\rightarrow \infty$ the shift quotient $r(x)\in \QQ(x)$ converges to some limit in $\mathbb{Q}$ or diverges to $\pm \infty$.  In case $r(x)$ diverges to $\pm \infty$ or converges to some $\ell \in \mathbb{Q}$ with $|\ell|>1$, then there exists $N\in\NN$ such that $\left|u_n\right|=\left|u_0\prod_{k=1}^{n}r(k)\right|>|t|$ for all $n\geqslant N$.  An analogous argument applies when $r(x)$ converges to some~$\ell$ with $|\ell|<1$.

The challenging case is when the shift quotient $r(x)\in \QQ(x)$ of the sequence converges to $\pm 1$ as $x\to\infty$. In this case, we establish the existence of the bound~$N$ by providing an infinite sequence $\langle p_i \rangle_{i=0}^{\infty}$ of integer  primes such that for all sufficiently large~$n$ one  prime in the sequence appears in the prime decomposition of~$u_n$ but not~$t$. We rely on results from analytic number theory on density of primes to construct the sequence $\langle p_i \rangle_{i=0}^{\infty}$ of primes.   The proof, presented in Section~\ref{sec:main-result}, is  constructive and allows to compute the bound~$N$.

\section{Connection to the Rohrlich-Lang conjecture}\label{sec:rohrlich}
In this section, we highlight the link between the Membership Problem
for hypergeometric sequences and the Rohrlich(-Lang) conjecture, which
concerns algebraic relations among 
values of the Gamma function at
rational points.

Let $\Gamma$ denote the Gamma function~\cite{KauersP11}.  Using
Euler's infinite-product characterisation, it
can be shown that certain infinite products of rational functions
converge to quotients of values of the Gamma function. In particular,
the following is standard, see for example~\cite{ChamberlandStraub13}.

\begin{proposition}\label{prop:limit_product_rat}
    Let $d\geqslant 1$ and $\alpha_1,\ldots,\alpha_d$ and $\beta_1,\ldots,\beta_d$ be nonzero complex numbers,
    none of which are negative integers. If $\alpha_1 + \ldots + \alpha_d = \beta_1 + \ldots + \beta_d$
    then
    \begin{equation}\label{eq:limit_product_rat}
        \prod_{k=1}^\infty \frac{(k+\alpha_1)\cdots(k+\alpha_d)}{(k+\beta_1)\cdots(k+\beta_d)}
            =\frac{\beta_1\cdots\beta_d}{\alpha_1\cdots\alpha_d} \frac{\Gamma(\beta_1)\cdots\Gamma(\beta_d)}{\Gamma(\alpha_1)\cdots\Gamma(\alpha_d)},
    \end{equation}
    otherwise the infinite product diverges.
\end{proposition}

We have already observed that the difficult case of the Membership Problem
for hypergeometric sequences is when the shift quotient $r(x)$
converges to $\pm 1$ as $x$ tends to infinity.  By splitting the numerator
and denominator of $r(x)$ into linear factors, the Membership Problem in this case is equivalent to
deciding whether there exists $n\in \NN$ such that the finite product
\begin{equation}\label{eq:finite_prod_rat}
    \prod_{k=1}^n \frac{(k+\alpha_1)\cdots(k+\alpha_d)}{(k+\beta_1)\cdots(k+\beta_d)}
\end{equation}
is equal to a given value $T := \frac{t}{u_0} \in\mathbb{Q}$. The link between this last
problem and \Cref{prop:limit_product_rat} arises from the fact that 
one can compute a bound $n_0$ such that for all $n>n_0$ the expression~\eqref{eq:finite_prod_rat} is either strictly
increasing or strictly decreasing as a function of $n$.  As we observe below, this allows us to 
reduce 
the Membership Problem 
to a question about infinite products:

\begin{proposition}\label{prop:membership_reduce_equal_gamma_quotient}
The Membership Problem for hypergeometric sequences
 \emph{with real parameters} reduces to
    deciding, given $d\geq 1$ and $\alpha_1,\ldots,\alpha_d, \beta_1,\ldots,\beta_d \in \RR \setminus \ZZ_{<0}$,
    whether
    \begin{equation}\label{eq:monomial_eq_gamma}
        \Gamma(\beta_1)\cdots\Gamma(\beta_d)=\Gamma(\alpha_1)\cdots\Gamma(\alpha_d).
    \end{equation}
\end{proposition}
\begin{proof}[Proof sketch]
As discussed in Section~\ref{subsec:shiftquo}, the only case of the Membership Problem that is not trivially decidable is when the shift quotient $r(x)\in \QQ(x)$ of the sequence $\langle u_n \rangle_{n=0}^\infty$ converges to $\pm 1$ as $x\to\infty$. We assume without loss of generality that for the instances of the problem we reason about below, the shift quotient is converging.
    We treat the case that the 
     product \eqref{eq:finite_prod_rat}
    is eventually strictly increasing.  The case
    that it is eventually strictly decreasing follows \emph{mutatis mutandis}.

Write $\omega := C \cdot \frac{\Gamma(\beta_1)\cdots\Gamma(\beta_d)}{\Gamma(\alpha_1)\cdots\Gamma(\alpha_d)}$ for the limit \eqref{eq:limit_product_rat}
of the finite product~\eqref{eq:finite_prod_rat}
as $n$ tends to~$\infty$, where $C = \frac{\beta_1\cdots\beta_d}{\alpha_1\cdots\alpha_d}$.
By the assumption that 
\eqref{eq:finite_prod_rat}
    is eventually strictly increasing, 
    we can 
compute $n_0$ such that 
$\frac{u_n}{u_0} < \omega$ for all $n > n_0$.
Let $T := \frac{t}{u_0}$.
If $\omega \leq T$ then 
we have $u_n\neq t$ for all $n>n_0$, 
and it remains to check by exhaustive search whether
$t\in \{u_0,\ldots,u_{n_0}\}$.
On the other hand, if $\omega > T$, then we can find $n_1 \geq n_0$ such that $u_n>t$ for all $n>n_1$.  Again, this leaves only a finite number of cases to check.

We thus need only decide whether or not $\omega \leq T$.  But
it is recursively enumerable whether $\omega < T$ and whether $\omega > T$, simply by computing 
$\omega$ to sufficient precision.  Thus
the Membership Problem for real parameters reduces to deciding whether $\omega =T$.
    Now note that $\Gamma\big(\frac{T}{C}+1\big)=\frac{T}{C}\Gamma\big(\frac{T}{C}\big)$, hence
    \[
        \omega=T
            \;\Leftrightarrow\; \frac{\Gamma(\beta_1)\cdots\Gamma(\beta_d)\Gamma\big(\frac{T}{C}\big)}{\Gamma(\alpha_1)\cdots\Gamma(\alpha_d)\Gamma\big(\frac{T}{C}+1\big)}=1.
    \]
    But the equation above is an instance of 
     \eqref{eq:monomial_eq_gamma} with two extra parameters, $\alpha_{d+1}:=\frac{T}{C}+1$ and 
     $\beta_{d+1}:=\frac{T}{C}$.  This completes the reduction.
\end{proof}

Unfortunately, deciding whether \eqref{eq:monomial_eq_gamma} holds 
appears to be a difficult problem, and we take a different approach to
solving the Membership Problem in the rest of this paper.
Nevertheless~\eqref{eq:monomial_eq_gamma} is still of interest for closely
related problems, such as the \emph{Threshold Problem} which asks, given
a hypergeometric sequence $\langle u_n \rangle_{n=0}^\infty$ and a
target $t$, whether $u_n\geqslant t$ holds for all $n\in\NN$.  In fact
we have:

\begin{proposition}\label{prop:positivity_intereduce_equal_gamma_quotient}
    The Threshold Problem for hypergeometric sequences with \emph{real parameters}
    is interreducible with the problem of
    deciding, given $d\geq 1$ and $\alpha_1,\ldots,\alpha_d,$  $\beta_1,\ldots,\beta_d \in \RR \setminus \ZZ_{<0}$,
    whether \eqref{eq:monomial_eq_gamma} holds.
\end{proposition}

Examining~\eqref{eq:monomial_eq_gamma} in more detail, we note that
the values of the Gamma function at rational points may be
transcendental.  Moreover, we are not aware of any lower bound on the
difference between the two terms of \eqref{eq:monomial_eq_gamma} in
case they are different, thus ruling out a numerical algorithm to
decide equality.  The following example from
\cite{ChamberlandStraub13} illustrates a case where a quotient is an
integer but in which none of the values of $\Gamma$ involved are known
to be algebraic:
\[
    \frac{\Gamma(\tfrac{1}{14})\Gamma(\tfrac{9}{14})\Gamma(\tfrac{11}{14})}
        {\Gamma(\tfrac{3}{14})\Gamma(\tfrac{5}{14})\Gamma(\tfrac{13}{14})}=2.
      \]
      
A different approach is to study the algebraic relations among
the values of the Gamma function.
D. Rohrlich considered the question of
giving a complete list of all \emph{multiplicative relations} relations among
the values of the Gamma function \emph{at rational points}. Recall that $\Gamma$ satisfies
the following three standard relations:
\begin{equation}\label{eq:gamma_std_rel}
    \begin{array}{@{}r@{\,}ll@{}}
    \Gamma(x+1)
        &=x\Gamma(x)
        &\text{(Translation)},\\
    \Gamma(x)\Gamma(1-x)
        &=\frac{\pi}{\sin(\pi x)}
        &\text{(Reflection)},\\
    \prod_{k=0}^{n-1}\Gamma\left(x+\frac{k}{n}\right)
        &=(2\pi)^{\tfrac{n-1}{2}}n^{\tfrac{1}{2}-nx}\Gamma(nx)
        &\text{(Multiplication)}
    \end{array}
\end{equation}
for all $x\in\CC$ except at poles. 

\begin{conjecture}[Rohrlich]
    Any multiplicative relation of the form
    \[
        \pi^{b/2}\prod_{a\in\QQ}\Gamma(a)^{m_a}\in\QQbar
    \]
    with $b$ and $m_a$ in $\ZZ$ is a consequence of the standard relations \eqref{eq:gamma_std_rel}.
\end{conjecture}

This conjecture was formalised by Lang in terms of ``universal
distribution'' \cite{Lang1977}.  This conjecture remains wide open and
is part of the larger work on the transcendence of periods
\cite{Waldschmidt06}.  The closest result to our problem is a theorem
by Koblitz and Ogus \cite{KoblitzOgus1979} giving a sufficient
condition under which a quotient of values of Gamma is algebraic.  A
stronger conjecture, known as the Rohrlich-Lang conjecture deals more
generally with polynomial relations. See
\cite[Section~24.6]{Andre2004} for more details on these conjectures.

Note that Rohrlich's conjecture is only concerned with rational
parameters and, as far as we are aware, there is no analog of this
conjecture for algebraic parameters.  We now observe that if
Rohrlich's conjecture is true, then the Membership and Threshold
Problems become decidable \emph{for rational parameters}. The decidability of the Threshold Problem for hypergeometric sequences subject to the Rohlich-Lang conjecture was first observed in the manuscript~\cite[Section 4]{kenison2020positivity}.

\begin{theorem}
  The Membership and Threshold Problems for hypergeometric sequences
  with rational parameters are decidable if Rohrlich's conjecture is
  true.
\end{theorem}
\begin{proof}[Proof sketch]
By \Cref{prop:membership_reduce_equal_gamma_quotient,prop:positivity_intereduce_equal_gamma_quotient}, the Membership and Threshold Problems both
reduce to the question of deciding equations of the form~\eqref{eq:monomial_eq_gamma}, in which
the parameters $\alpha_i$ and $\beta_i$ are rational.  Assuming
Rohrlich's conjecture, the latter problem is recursively enumerable:
if equality holds, then the equation has a finite derivation using the
standard relations~\eqref{eq:gamma_std_rel}.  On the other hand, the problem is also
straightforwardly co-recursively enumerable: if Equation~\eqref{eq:monomial_eq_gamma} does not
hold, then by computing the left and right-hand sides to sufficient
precision we will eventually conclude that the two terms are not equal.
\end{proof}

\section{The Main Result} \label{sec:main-result}

\begin{figure*}
\scalebox{.75}{
\begin{tikzpicture}
\def \x {1/2}

\draw[gray, ultra thick] (0,4) -- (17*\x,4);  
\node[left] at (33*\x,4) {{\bf prime $17$}};
\foreach \y in {1,...,16}
\draw[draw=gray] (\y*\x,4.1)--(\y*\x,3.9);
 \draw[gray, line width=1mm, opacity=.5] (0,4.2) arc[start angle=150, end angle=210, radius=.4cm];
\draw[gray, line width=1mm, opacity=.5] (17*\x,4.2) arc[start angle=30, end angle=-30, radius=.4cm];
\node[left] at (.2,3.65) {{\large  $0$}};
\filldraw [violet!30!white, opacity=0.5] (3*\x,4.1) rectangle (4*\x,3.9);
\filldraw [violet!30!white, opacity=0.5] (5*\x,4.1) rectangle (16*\x,3.9);
\node at (9*\x,4.3) {\textcolor{violet}{expanding}};

\draw[gray, ultra thick] (0,2) -- (23*\x,2);
\node[left] at (33*\x,2) {{\bf prime $23$}};
\foreach \y in {1,...,22}
\draw[draw=gray] (\y*\x,2.1)--(\y*\x,2-.1);
\draw[gray, line width=1mm, opacity=.5] (0,2.2) arc[start angle=150, end angle=210, radius=.4cm];
\draw[gray, line width=1mm, opacity=.5] (23*\x,2.2) arc[start angle=30, end angle=-30, radius=.4cm];
\node[left] at (.2,2-.35) {{\large  $0$}};
\filldraw [violet!30!white, opacity=0.5] (6*\x,2.1) rectangle (7*\x,2-.1);
\filldraw [violet!30!white, opacity=0.5] (8*\x,2.1) rectangle (22*\x,2-.1);
\node at (13*\x,2.3) {\textcolor{violet}{expanding}};
\draw[thick, blue, draw=none, pattern=north east lines] (8*\x,3.5) -- (16*\x,3.5) -- (16*\x,2.5) -- (8*\x,2.5) -- cycle;
\node at (12*\x,3) {contiguous};

\draw[gray, ultra thick] (0,0) -- (29*\x,0);
\node[left] at (33*\x,0) {{\bf prime $29$}};
\foreach \y in {1,...,28}
\draw[draw=gray] (\y*\x,.1)--(\y*\x,-.1);
\draw[gray, line width=1mm, opacity=.5] (0,.2) arc[start angle=150, end angle=210, radius=.4cm];
\draw[gray, line width=1mm, opacity=.5] (29*\x,.2) arc[start angle=30, end angle=-30, radius=.4cm];
\node[left] at (.2,-.35) {{\large  $0$}};
\filldraw [violet!30!white, opacity=0.5] (9*\x,.1) rectangle (10*\x,-.1);
\filldraw [violet!30!white, opacity=0.5] (11*\x,.1) rectangle (28*\x,-.1);
\node at (18*\x,.3) {\textcolor{violet}{expanding}};
\draw[thick, blue, draw=none, pattern=north east lines] (11*\x,1.5) -- (22*\x,1.5) -- (22*\x,0.5) -- (11*\x,0.5) -- cycle;
\node at (16.5*\x,1) {contiguous};

\node[left] at (3*\x+.2,3.7) {{\large  $3$}};
\node[teal]  at (3*\x,4.3) {{\scriptsize $\beta_1$}};
\filldraw[teal] (3*\x,4) circle (3pt);

\node[left] at (4*\x+.2,3.7) {{\large  $4$}};
\node[green!70!black]  at (4*\x,4.3) {{\scriptsize $\alpha_1$}};
\filldraw[green!70!black] (4*\x,4) circle (3pt);

\node[left] at (5*\x+.2,3.7) {{\large  $5$}};
\node[green!70!black]  at (5*\x,4.3) {{\scriptsize $\alpha_2$}};
\filldraw[green!70!black] (5*\x,4) circle (3pt);

\node[left] at (6*\x+.2,3.7) {{\large  $6$}};
\node[green!70!black]  at (6*\x,4.3) {{\scriptsize $\alpha_3$}};
\filldraw[green!70!black] (6*\x,4) circle (3pt);

\node[left] at (13*\x+.2,3.7) {{\large  $13$}};
\node[teal]  at (13*\x,4.3) {{\scriptsize $\beta_2$}};
\filldraw[teal] (13*\x,4) circle (3pt);

\node[left] at (16*\x+.2,3.7) {{\large  $16$}};
\node[teal]  at (16*\x,4.3) {{\scriptsize $\beta_3$}};
\filldraw[teal] (16*\x,4) circle (3pt);

\node[left] at (6*\x+.2,1.7) {{\large  $6$}};
\node[teal]  at (6*\x,2.3) {{\scriptsize $\beta_1$}};
\filldraw[teal] (6*\x,2) circle (3pt);

\node[left] at (7*\x+.2,1.7) {{\large  $7$}};
\node[green!70!black]  at (7*\x,2.3) {{\scriptsize $\alpha_1$}};
\filldraw[green!70!black] (7*\x,2) circle (3pt);

\node[left] at (8*\x+.2,1.7) {{\large  $8$}};
\node[green!70!black]  at (8*\x,2.3) {{\scriptsize $\alpha_2$}};
\filldraw[green!70!black] (8*\x,2) circle (3pt);

\node[left] at (9*\x+.2,1.7) {{\large  $9$}};
\node[green!70!black]  at (9*\x,2.3) {{\scriptsize $\alpha_3$}};
\filldraw[green!70!black] (9*\x,2) circle (3pt);

\node[left] at (19*\x+.2,1.7) {{\large  $19$}};
\node[teal]  at (19*\x,2.3) {{\scriptsize $\beta_2$}};
\filldraw[teal] (19*\x,2) circle (3pt);

\node[left] at (22*\x+.2,1.7) {{\large  $22$}};
\node[teal]  at (22*\x,2.3) {{\scriptsize $\beta_3$}};
\filldraw[teal] (22*\x,2) circle (3pt);

\node[left] at (9*\x+.2,-.3) {{\large  $9$}};
\node[teal]  at (9*\x,.3) {{\scriptsize $\beta_1$}};
\filldraw[teal] (9*\x,0) circle (3pt);

\node[left] at (10*\x+.2,-.3) {{\large  $10$}};
\node[green!70!black]  at (10*\x,.3) {{\scriptsize $\alpha_1$}};
\filldraw[green!70!black] (10*\x,0) circle (3pt);

\node[left] at (11*\x+.2,-.3) {{\large  $11$}};
\node[green!70!black]  at (11*\x,.3) {{\scriptsize $\alpha_2$}};
\filldraw[green!70!black] (11*\x,0) circle (3pt);

\node[left] at (12*\x+.2,-.3) {{\large  $12$}};
\node[green!70!black]  at (12*\x,.3) {{\scriptsize $\alpha_3$}};
\filldraw[green!70!black] (12*\x,0) circle (3pt);

\node[left] at (25*\x+.2,-.3) {{\large  $25$}};
\node[teal]  at (25*\x,.3) {{\scriptsize $\beta_2$}};
\filldraw[teal] (25*\x,0) circle (3pt);

\node[left] at (28*\x+.2,-.3) {{\large  $28$}};
\node[teal]  at (28*\x,0.3) {{\scriptsize $\beta_3$}};
\filldraw[teal] (28*\x,0) circle (3pt);

\end{tikzpicture}
}
	\caption{Consider the sequence $\langle w_n \rangle_{n=0}^{\infty}$ given in Example~\ref{ex:p1}. Observe that 
$\beta_1 \, \preceq_{p} \, \alpha_1 \, \preceq_{p}\, \alpha_2  \, \preceq_{p} \, \alpha_3 \, \preceq_{p} \, \beta_2 \, \preceq_{p}\, \beta_3$ for all primes~$p\in \{17,23,29\}$.
The two families of intervals~$\overline{\rm (\beta_1,\alpha_1)}$ and $\overline{\rm (\alpha_2,\beta_3)}$
	are $s$-unbalanced, where only the latter 
	  is $s$-expanding.
	In particular, the distance between  residues of $\alpha_2$ and $\beta_3$ modulo $23$ is greater than their respective distance modulo $17$. The same holds for their distance modulo $29$ compared to their distance modulo $23$. The $s$-expanding $s$-unbalanced intervals for $17$, $23$ and $29$ are contiguous, which in turn ensures that for all $n\in \{5,\ldots, 27\}$ either 
	$17$ or $23$ or $29$ divides $w_n$. 
 See Example~\ref{ex:p1} for a more detailed discussion. }
	\label{fig:example}
\end{figure*}
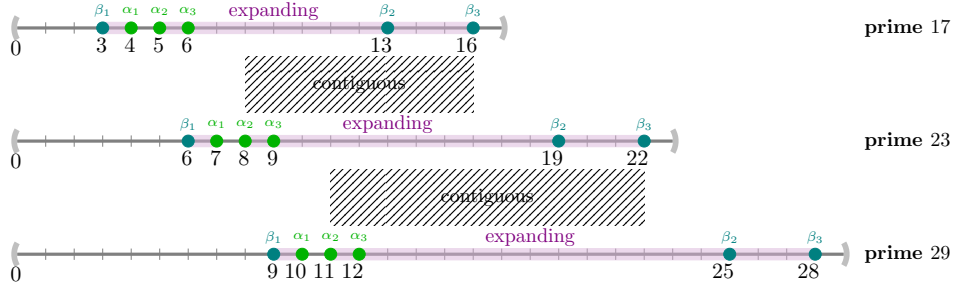

\paragraph{The ring~$\ZZ_{(p)}$.}

Let $p$ be a prime.
We denote by $v_p : \QQ \to \ZZ \cup \{\infty\}$ the $p$-adic
valuation on $\QQ$.  Recall that for a non-zero rational number~$x$,
the valuation~$v_p(x)$ is the unique integer such that 
$x$ can be written in the form
$x=p^{v_p(x)} \frac{a}{b}$ with $p \not \mid ab$. Following the standard convention,  define
$v_p(0):=\infty$.

We denote by $\ZZ_{(p)}$ the 
ring $\{ x \in \mathbb{Q} : v_p(x)\geq 0\}$.
Alternatively, we have
 $\ZZ_{(p)}=\{ \frac{a}{b} : a,b\in \mathbb{Z}, p \not \mid b\}$.
This is a local ring, whose unique maximal ideal is
the principal ideal $p\ZZ_{(p)}$.  The quotient
$\ZZ_{(p)}/p\ZZ_{(p)}$ is isomorphic to the finite field
$\mathbb{F}_p$.  Specifically, we consider the quotient map
$\rem_p : \ZZ_{(p)} \to \mathbb{F}_p$, given by
$$ \rem_p\left(\frac{a}{b}\right) := ab^{-1} \bmod p \, . $$
Henceforth, when we say that  $\frac{a}{b}\in \ZZ_{(p)}$ has $p$
as a prime divisor we refer to divisibility in $\ZZ_{(p)}$.


\subsection{Overview of the proof}
In this section we give a high-level overview of our main result.

As discussed in Section~\ref{subsec:shiftquo}, to prove
decidability of MP it remains to handle those instances in which the shift
quotient $r(x)\in \QQ(x)$ converges to $\pm 1$ as
$x\to\infty$.  In such instances, $r(x)$ is necessarily the quotient of two polynomials of equal degree.
Throughout this section, we fix such an instance of MP, consisting of a
 rational value~$t$ and a hypergeometric sequence $\langle u_n \rangle_{n=0}^\infty$  with rational parameters and rational initial value $u_0$.
                                                                                                      
We aim at computing a bound~$N$ such that all indices of $t$ in $\langle u_n \rangle_{n=0}^\infty$ are at most~$N$.  Our strategy is to find $N$ such 
that for all~$n > N$ there exists a prime
$p$ that is a prime divisor of~$u_n$ but not of $t$.
To explain this idea in more detail,
rewrite the shift quotient~$r(x)$ as
\[\frac{(x-\alpha_1)\cdots (x-\alpha_d)}{(x-\beta_1)\cdots (x-\beta_d)}\, , \]
 where the $\alpha_i$ and the $\beta_i$ are in~$\QQ\setminus \ZZ_{\geq 0}$.
We denote by $\As$ the multiset $\{\alpha_1,\ldots,\alpha_d\}$ consisting of all the (possibly repeated) roots of the numerator and by $\Bs$ the multiset $\{\beta_1,\ldots,\beta_d\}$ of the roots of the denominator.
Denote by $\Supp(C)$ the underlying set of a multiset~$C$.
For each element $x\in\Supp(C)$ we write $m_{C}(x)$
for its multiplicity in $C$.
We write $\As \uplus \Bs$ for the multiset with underlying set $\Supp(\As)\cup\Supp(\Bs)$ where the multiplicity of each of its elements~$x$  is $m_{\As}(x)+m_{\Bs}(x)$. 

Given a prime $p$, we have that
$$v_p(u_n) = v_p(u_0) + v_p\Big(\prod_{k=1}^{n}r(k)\Big)$$
for all $n \in \mathbb{N}$.  In particular,
if $v_p(t)=v_p(u_0)=0$ and 
$$
v_p\Big(\prod_{k=1}^{n}r(k)\Big) \neq 0 \, ,
$$
then $u_n \neq t$.
Furthermore the term 
$v_p\Big(\prod_{k=1}^{n}r(k)\Big) $
can be expanded as
\begin{equation}
\label{eq:summ}
 S_p (n) := \sum_{k = 1}^{n} \left( \sum_{\alpha\in \As} v_p(k-\alpha) - \sum_{\beta \in \Bs} v_p(k-\beta) \right)\, .
 \end{equation}
Thus, for all $n\in \mathbb{N}$ and all primes $p$ not dividing $u_0$ or $t$,  if $S_p(n)\neq 0$
then $u_n \neq t$.

\paragraph{The preorder $\preceq_r$.}
Given an integer prime~$p$, we define a preorder~$\preceq_p$ on~$\ZZ_{(p)}$ by writing
$\frac{a}{b}\preceq_p \frac{a'}{b'}$ if and only if $\rem_p(\frac{a}{b}) \leq \rem_p(\frac{a'}{b'})$, where $\leq $
is the usual order on~$\{0,\ldots,p-1\}$.

Denote by $b$ the least common denominator of all fractions in~$\As \uplus \Bs$. 
For every prime $p$ in  
 the arithmetic progression $b\NN+1$, all elements of
$\As \uplus \Bs$ are in $\ZZ_{(p)}$.
 In~\Cref{prop:order-neg} we show that for all sufficiently large 
 primes~$p\in b\NN+1$  the orders~$\preceq_p$ restricted to  
$\As \uplus \Bs$ are identical. 
We denote this common preorder by $\preceq_r$, where $r$ is 
the shift quotient of our fixed sequence.
 
 \paragraph{Unbalanced intervals}
Given an integer prime~$p$,
let $\frac{a}{b} \in \ZZ_{(p)}$ be such that $gcd(a,b)=1$. Note that 
$v_p(b)=0$, and for all $k\in \{1,\ldots,p-1\}$ such that $0<|kb-a|<p^2$ we have
\begin{equation}\label{eq:vprm}
v_p\left(k-\frac{a}{b}\right)=\begin{cases} 1 & \text{if }\rem_p (\frac{a}{b}) =
\rem_p(k),\\ 0 & \text{otherwise.}  \end{cases} \end{equation}

Recall that $\As \uplus \Bs \subseteq  \QQ\setminus \ZZ_{\geq 0}$, and assume that all its elements are given in a reduced form (i.e., $\gcd(a,b)=1$ for all $\frac{a}{b} \in \As \uplus \Bs$).
Let $p$ be a prime such that all elements of $\As \uplus \Bs$ are in $\ZZ_{(p)}$.  
Let~$n\in \{1,\ldots,p-1\}$ be such that
for all $k\in \{1,\ldots,n\}$, for all~$\frac{a}{b} \in \As \uplus \Bs$,
the inequalities 
$0<|kb-a|<p^2$ hold. 
In this case, by Equation~\eqref{eq:vprm}, 
the $p$-adic valuations in Equation~\eqref{eq:summ} 
all take value $0$ or $1$.
This means that $S_p(n)$  is non-zero
if and only if
\begin{equation} \label{eq:ineqSn}
	|\{\alpha\in \As : \alpha \preceq_p n\}| \;\neq\; |
\{\beta\in \Bs : \beta \preceq_p n \}|.
\end{equation}

We say that $n\in\mathbb{N}$ is
\emph{$p$-unbalanced} if $S_p(n)\neq 0$.  We extend this notion to 
sub-intervals of $\NN$ by saying that an interval $I\subseteq \mathbb{N}$ is \emph{$p$-unbalanced} if  all $n\in I$ are $p$-unbalanced.
By \Cref{eq:ineqSn}, the maximal $p$-unbalanced sub-intervals of~$\{1,\ldots,p-1\}$ will have endpoints that are equal or
adjacent to the images $\rem_p(\frac{a}{b})$ of~$\frac{a}{b}\in \As \uplus \Bs$; see~\Cref{ex:p1}. 

Let $\gamma$ and $\gamma'$ be distinct elements of $A\uplus B$ such 
that $\gamma  \prec_r \gamma'$.  We denote by 
$\overline{\rm (\gamma,\gamma')}$ the family of sub-intervals 
$\{n\in \NN : \rem_p(\gamma) \leq n < \rem_p(\gamma')\}$
of $\NN$
indexed by primes $p \in b\NN+1$ whose respective orders agree with $\prec_r$
(i.e., indexed by sufficiently large primes in $b\NN+1$).

We show that for such a family of sub-intervals, indexed by primes $p$, either every interval is $p$-unbalanced or none of the intervals is $p$-unbalanced.
In the former case we say that the family of intervals is $r$-unbalanced, 
where $r$ is the shift quotient.

\paragraph{Expanding families and contiguous intervals.}
Let $\gamma,\gamma' \in$ 
$\As \uplus \Bs$
such that
$\gamma\prec_r\gamma'$. 
In~\Cref{prop:p_adic_distance} we  prove that for a sufficiently large prime $p$ belonging to the arithmetic progression $b\NN+1$, the \emph{distance} $\rem_p(\gamma')-\rem_p(\gamma)$ between the respective residues of $\gamma$ and $\gamma'$ modulo $p$ is a strictly increasing function of~$p$, if and only if $\gamma-\gamma' \notin\ZZ$.  
In this case we say that the family of intervals
$\overline{\rm (\gamma,\gamma')}$ is \emph{$r$-expanding}.

The identification
of  expanding families of unbalanced intervals
is a crucial element in our proof. 
We further show that: 

\begin{proposition}\label{prop:reduction}
	Given  $r(x)\in \QQ(x)$ converging to $\pm 1$ as $x\to\infty$, either
	\begin{enumerate}
		\item there exists an $r$-expanding $r$-unbalanced family of intervals,
		or
\item  otherwise, every hypergeometric sequence $\langle u_n \rangle_{n=0}^\infty$ with shift quotient $r(x)$ is a rational function of $n$. 
	\end{enumerate}
\end{proposition}


For the case when  $\langle u_n \rangle_{n=0}^\infty$ is a rational function of $n$, we can rewrite it as $u_n=\frac{f(n)}{g(n)}$ with $f,g\in \QQ[x]$. In order to test  membership of~$t$, it  suffices to check  whether the polynomial $f(x)-tg(x)$ has an integer root. For more details and the proof of \Cref{prop:reduction}, see \Cref{appendix:reduction}.
Henceforth, we assume without loss of generality that there exists 
an $r$-expanding $r$-unbalanced family of intervals for our fixed instance of~MP.

Pick $\gamma, \gamma' \in \As \uplus \Bs$ such that  $\overline{\rm (\gamma,\gamma')}$ is  
an $r$-expanding $r$-unbalanced family of  intervals.
Let $p,q \in b\NN+1$ be primes sufficiently large that their respective orders agree with $\prec_r$. In \Cref{prop:constructqfromp} we show that if $p< q < p(1 + \frac{1}{b})+C$ with $C$ a constant depending only on $r$, the respective intervals between $\gamma$ and $\gamma'$ for primes $p$ and $q$ are \emph{contiguous}. That is, we show that $\rem_{q}(\gamma) \leq \rem_{p}(\gamma')$ where $\leq$ is the usual order on~$\ZZ$.
 By previous arguments,  for all~$n$ with
 $\rem_p(\gamma) \leq n < \rem_{q}(\gamma')$ either $v_p(u_n) \neq 0$ or $v_q(u_n) \neq 0$.

%

From the instance of MP we now construct an $r$-expanding $r$-unbalanced 
family of intervals that covers $\{ n \in \NN : n > N\}$
for some effectively computable finite bound $N$.
This will conclude our conceptually simple proof for decidability of MP.

\paragraph{An infinite sequence of primes with contiguous 
unbalanced intervals.} 
We use effective bounds on the density of primes in arithmetic progressions to  construct a sequence $\langle p_i \rangle_{i = 0}^{\infty}$ of primes with 
contiguous $r$-unbalanced intervals. 
 Intuitively speaking, given a prime $p_i \in b\mathbb{N}+1$, we would need  $p_{i+1} <  p_i(1 + \frac{1}{b})+C$ with $C$ a fixed constant depending on~$r$. We prove that if $p_i$ is large enough there always exists another prime $p_{i+1}\in b\mathbb{N}+1$  where $p_{i+1} < p_i (1+\frac{1}{b})+C$. See \Cref{prop:primes_progression} for more details.

\begin{example}\label{ex:p1}
Consider the sequence $\langle w_n \rangle_{n=0}^\infty$ 
defined by $w_0=1$ and the shift quotient 
\[s(x):=\frac{(x+\frac{9}{2})(x+\frac{7}{2})(x+\frac{5}{2})}{(x+\frac{11}{2})(x+4)(x+1)}.\]
The rational function~$s(x)$ converges to~$1$ from above as $x \to \infty$. This implies that the sequence $\langle w_n \rangle_{n=0}^\infty$ is 
monotonically increasing to its limit value, that is 
\[
        \prod_{k=0}^\infty \frac{(k+\frac{9}{2})(k+\frac{7}{2})(k+\frac{5}{2})}{(k+\frac{11}{2})(k+4)(k+1)}
            =\frac{\Gamma(\frac{11}{2})\Gamma(4)\Gamma(1)}{\Gamma(\frac{9}{2})\Gamma(\frac{7}{2})\Gamma(\frac{5}{2})}=\frac{3\cdot 2^5}{5\pi}.
\] 

We give two arguments that~$\frac{13}{6}$ does not lie in the sequence. 
First, using the fact $\langle w_n\rangle_{n=0}^\infty$ is strictly increasing,
it suffices to observe that $w_6 > \frac{13}{6}$ and that none of $w_0,\ldots,w_5$
equals $\frac{13}{6}$.
Such an argument is possible because $w_n$ does not converge to~$\frac{13}{6}$. 
Next, we prove the non-membership of~$\frac{13}{6}$ in the sequence using our approach based on prime divisors of the $w_n$. 
To this end, let 
\begin{align*}
	\alpha_1:=\frac{-9}{2}  &\quad \quad & \alpha_2:= \frac{-7}{2} &\quad \quad &  \alpha_3:=\frac{-5}{2}\\
	 \beta_1:=\frac{-11}{2}  &\quad \quad &  \beta_2:= -4 &\quad \quad &  \beta_3:=-1
\end{align*}

Considering that  $v_{13}(\frac{13}{6})=1$, in our approach,  we  use  larger primes  to rule out the membership of~$\frac{13}{6}$.
For the prime~$17$, as depicted in~\Cref{fig:example}, we have  that  
\[\beta_1 \, \preceq_{17} \, \alpha_1 \, \preceq_{17}\, \alpha_2  \, \preceq_{17} \, \alpha_3 \, \preceq_{17} \, \beta_2 \, \preceq_{17}\, \beta_3.\]
The maximal $17$-unbalanced intervals   in our example are $\{3\}$ and $\{5,6,\ldots, 15\}$.
This implies that, for $n\in \{1,\ldots,16\}$,  $S_{17}(n)$ is non-zero  if and only if $n$ belongs to~$\{3\} \cup \{5,\ldots, 15\}$.

As it turns out $17$ is  sufficiently large so that, for all primes $p\geq 17$,  the respective order $\prec_p$   agrees with $\prec_s$. Consequently,
the families of intervals 
$\overline{\rm (\beta_1,\alpha_1)}$ and $\overline{\rm (\alpha_2,\beta_3)}$
are $s$-unbalanced.
Furthermore, the family  
$\overline{\rm (\alpha_2,\beta_3)}$ is $s$-expanding, whereas
$\overline{\rm (\beta_1,\alpha_1)}$  is not. As shown in Figure~\ref{fig:example},
the  intervals between  $\alpha_2$ and $\beta_3$ are contiguous for primes in $\{17, 23,29\}$.
This, in turn, ensures that for all $n\in \{5,\ldots, 27\}$ either 
	$17$, $23$, or $29$ divides $w_n$. 

By the main theorem in~\cite{Nagura}, for all primes $p\geq 8$, there exists another prime less than $p(1+\frac{1}{2})$. This, in combination with the above, guarantees that for all $n\geq 5$, there exists a prime
$p$ other than $13$ appearing in the factorisation of~$w_n$. This reduces the membership test of $\frac{13}{6}$ to the finite set $\{w_1,\ldots, w_4\}$.
\end{example}

\subsection{Technical Lemmas}

Recall that the elements in $\As \subseteq \QQ\setminus\ZZ_{\geq 0}$ are the roots of the numerator, and the elements in $\Bs \subseteq \QQ\setminus\ZZ_{\geq 0}$ are the roots of the denominator of the shift quotient $r(x)$ of the fixed sequence $\langle u_n \rangle_{n=0}^\infty$.

Let $b>0$ be the least common denominator of the fractions in~$\As \uplus \Bs$. 
Then every element 
$\gamma \in \As \uplus \Bs$
 admits a unique representation in the form $\gamma=c-\frac{a}{b}$ under the conditions
$c \in \mathbb{Z}$ and $a \in \{1,\ldots,b\}$.
We call this the \emph{canonical representation}.
Associated with this, define a nonnegative integer
$$ N_\gamma := \left\{
    \begin{array}{ll} \max\left(c,\lceil\frac{bc}{b-a}\rceil\right) & \mbox{ if $c\geq 1$,} \\
      \max\left(-c,\lceil\frac{-bc}{a}\rceil\right) & \mbox{ if $c \leq 0$.}
                                          \end{array}\right . $$
Note that $N_\gamma$ is well-defined, i.e., there is no division by zero.
Indeed, if $c\geq 1$, by the convention that $\gamma \not\in \mathbb{Z}_{\geq 0}$, the value $b-a$ is non-zero, whereas  $a\neq 0$ ensures well definedness in the second case.

\begin{proposition}
\label{prop:representation}
Let $\gamma=c-\frac{a}{b}$ be a canonical representation.
Then for all primes $p > N_\gamma$ such that $p \in b\mathbb{N}+1$
we have
$\rem_p(\gamma) = c +\frac{(p-1)a}{b}$.
\end{proposition}
\begin{proof}
  The assumption that $p \in b\mathbb{N}+1$ implies that
  $\rem_p\left(-\frac{1}{b}\right)=\frac{p-1}{b}$.
  Thus, by the homomorphism property of $\rem_p$,
  it only remains to verify that $c+\frac{(p-1)a}{b} \in \{0,\ldots,p-1\}$.  To this end, there are two cases, following the definition of $N_\gamma$.

The first case is that $c \geq 1$.  Here we clearly have $c+\frac{(p-1)a}{b} \geq 0$.  Furthermore, by the assumption $p > N_\gamma$, we have $p-1 \geq \frac{bc}{b-a}$.  Recall also that $a\neq b$ due to the assumption that $\gamma \not\in \mathbb{Z}_{\geq 0}$.  Thus, multiplying the previous inequality by $b-a>0$, we have $(b-a)(p-1) \geq bc$.  Dividing by $b$ and rearranging terms, we conclude that $c + \frac{(p-1)a}{b} \leq p-1$.

The second case is that $c \leq 0$.  Here, since $a \leq b$, it is clear that $c+\frac{(p-1)a}{b} \leq p-1$.  Furthermore, by the assumption $p > N_\gamma$, we have $p-1 \geq \frac{-bc}{a}$.  Multiplying the latter inequality by $\frac{a}{b}$ and rearranging terms, we get $c+\frac{(p-1)a}{b} \geq 0$.

\end{proof}

Next we use \Cref{prop:representation} to show that, given distinct
$\gamma$ and $\gamma'$ in $\As \uplus \Bs$,
 for all large enough primes $p$ in the arithmetic progression $b\NN + 1$, their order respective to $\prec_p$ is fixed.

\begin{proposition}\label{prop:order-neg}
Let $\gamma=c-\frac{a}{b}$ and $\gamma'=c'-\frac{a'}{b}$  be canonical representations.
For all primes $p > b(N_{\gamma}+ N_{\gamma'}) + 1$ such that $p \in b\mathbb{N}+1$ we have:
\[ \gamma \prec_p \gamma' \text{ if and only if } ((a < a') 
	\text{ or } (a=a' \text{ and } c<c')). \]
\end{proposition} 

\begin{proof}
For the first direction, assume that $\gamma \prec_p \gamma'$. Following \Cref{prop:representation}, since $p > N_{\gamma}+ N_{\gamma'}$ and $p\in b\NN + 1$, we can rewrite the assumption as $c + \frac{(p-1)a}{b} < c' + \frac{(p-1)a'}{b}$. We can rearrange the inequality to obtain
\[ cb - c'b + a' - a < p(a' - a). \]
Towards a contradiction, assume that  $a > a'$. In this case, the above yields  $\frac{c - c'}{a' - a}b + 1 > p$. Since $N_{\gamma} \geq |c|$ and $N_{\gamma'} \geq |c'|$, by the  assumption that $p > b(N_{\gamma}+ N_{\gamma'}) + 1$ we have that $p > (c - c')b + 1 > \frac{c - c'}{a' - a}b + 1$, a contradiction. Again, towards a contradiction, assume that $a = a'$ but $c \geq c'$. Since $b>0$ we can multiply the inequality by $b$ to get $cb \geq c'b$, a contradiction. The claim follows.
 
To show the other direction of the equivalence, we need to look at two cases depending on $a$ and $a'$.
First assume that $a < a'$. By $p > b(N_{\gamma}+ N_{\gamma'}) + 1$ we can write
\[
 p > (c - c')b + 1 \geq \frac{c - c'}{a' - a}b + 1 = \frac{cb - c'b + a' - a}{a' - a}. 
\]
Since $a' - a > 0$, we can multiply the above inequality by $(a'-a)$ to obtain
 $p(a' - a) > cb - c'b + a' - a$. 
By rearranging the terms, we get
$(p-1)a + cb < (p-1)a' + c'b$.  
As $p > N_{\gamma}+ N_{\gamma'}$ and $p\in b\NN +1$ by \Cref{prop:representation}, it follows that $\gamma \prec_p \gamma'$.

For the second case, assume $a = a'$ and $c < c'$. Since $b>0$, we can write $cb < c'b $. It follows that
$(p-1)a + cb < (p-1)a + c'b$. 
Again, since  $p > N_{\gamma} + N_{\gamma'}$ and $p\in b\NN +1$ by \Cref{prop:representation}, it follows that $\gamma \prec_p \gamma'$.
\end{proof}

Associated to the shift quotient~$r(x)$ of the sequence~$\langle u_n \rangle_{n=0}^\infty$, define the nonnegative integer 
\begin{equation}\label{eq:nr}
	N_r:=b\sum_{\gamma\in \As\uplus \Bs} N_{\gamma}+1.
\end{equation}
  
From \Cref{prop:order-neg} it follows that for all primes $p >N_r$ in the arithmetic progression $b\NN + 1$
\begin{itemize}
	\item $\rem_p(\gamma)$, $\rem_p(\gamma')$ are distinct for all distinct $\gamma, \gamma' \in \As \uplus \Bs$,
	\item the orders $\preceq_p$ on $\As \uplus \Bs$ are identical.
\end{itemize}

We henceforth denote by $\preceq_r$ the common order on 
$\As \uplus\Bs$
for all primes $p >N_r$ in the arithmetic progression $b\NN + 1$.
%
Let $\gamma,\gamma' \in \As \uplus \Bs$.
In the following proposition we show that  for larger and larger primes $p\in b\NN+1$ whose orders agree with $\prec_r$, the distance between  $\rem_p (\gamma)$ and $\rem_p (\gamma')$ gets larger and larger
if and only if $\gamma - \gamma' \notin \ZZ$.

%

\begin{proposition}\label{prop:p_adic_distance}
Let $p, q \in b\mathbb{N}+1$ be primes with $q> p> N_r$.
Let $\gamma,\gamma' \in \As \uplus \Bs$ be such that 
 $\gamma \prec_r \gamma'$. Then $\gamma - \gamma' \notin \ZZ$ if and only if 
 	\[ \rem_{p}(\gamma') - \rem_{p}(\gamma) < \rem_{q}(\gamma') - \rem_{q}(\gamma)\]
	where $<$ is the total order on~$\ZZ$.
\end{proposition} 

\begin{proof}
For the first direction, assume that $\gamma - \gamma' \notin \ZZ$. Given the canonical representations $\gamma = c - \frac{a}{b}$ and $\gamma' = c' - \frac{a'}{b}$, this implies that $a \neq a'$. Now since $\gamma \prec_r \gamma'$, by \Cref{prop:order-neg} we have $a<a'$.

Since $p < q$ and $b > 0$, we can write
$\frac{p - 1}{b} < \frac{q - 1}{b}$. 
We can multiply the inequality by $(a'-a)$ to obtain
\[ \frac{p - 1}{b}(a' - a) < \frac{q - 1}{b}(a' - a).
			\]
Now by adding $(c' - c)$ on both sides of the above inequality we get
\[ c' + \frac{p - 1}{b}a' - \left(c + \frac{p - 1}{b}a \right) < 
				c' + \frac{q - 1}{b}a' - \left(c + \frac{q - 1}{b}a \right).
			\]
By the assumption that $p, q \in b\mathbb{N}+1$ with $p, q > N_r$, we can use \Cref{prop:representation} to conclude that
\[ \rem_{p}(\gamma') - \rem_{p}(\gamma) < \rem_{q}(\gamma') - \rem_{q}(\gamma).\]

For the other direction, assume that $\rem_{p}(\gamma') - \rem_{p}(\gamma) < \rem_{q}(\gamma') - \rem_{q}(\gamma)$. Towards a contradiction, assume furthermore that $\gamma - \gamma' \in \ZZ$.
Using the canonical representations $\gamma = c - \frac{a}{b}$ and $\gamma' = c' - \frac{a'}{b}$, by \Cref{prop:representation} we can rewrite the inequality as:
\[ \left(c' + \frac{p - 1}{b}a'\right) - \left(c + \frac{p - 1}{b}a \right) < \left(c' + \frac{q - 1}{b}a'  \right) -  \left(c + \frac{q - 1}{b}a \right).
\]
The inequality simplifies to 
\[ \frac{p - 1}{b}(a'-a) < \frac{q - 1}{b}(a' - a).
\]
The assumption $\gamma - \gamma' \in \ZZ$ implies that $a = a'$, which with the above gives $0 < 0$, a contradiction.
\end{proof}	


Let $\gamma, \gamma' \in \As \uplus \Bs$ be such that  $\overline{\rm (\gamma,\gamma')}$ is  
an $r$-expanding $r$-unbalanced family of intervals.
For a prime $p \in b\mathbb{N}+1$  with $p> N_r$,  we  obtain a condition on larger primes
$q\in b\NN+1$  that ensures the 
 intervals  $\{n \in \NN : \rem_p(\gamma) \leq n < \rem_p(\gamma') \}$ and $\{n \in \NN : \rem_q(\gamma) \leq n < \rem_q(\gamma') \}$ are contiguous.


\begin{proposition}\label{prop:constructqfromp}
Let $p, q \in b\mathbb{N}+1$ be primes with $q> p> N_r$.
Let $\gamma,\gamma' \in \As \uplus \Bs$ such that $\gamma - \gamma' \notin \ZZ$ and 
 $\gamma \prec_r \gamma'$.
 We have 
	$\rem_{q}(\gamma) < \rem_{p}(\gamma')$ if 
	 \[q<p+\frac{p-1}{b}+ C\] 
	for some effective constant $C$ depending on $\gamma$ and $\gamma'$.
\end{proposition}
\begin{proof}
Given the canonical representations $\gamma = c - \frac{a}{b}$ and $\gamma' = c' - \frac{a'}{b}$,
write $C := \frac{c' - c}{a}b$. Then the initial assumption can be rewritten as
\[ q < p+\frac{p-1}{b}+ \frac{c' - c}{a}b. \]
We can further rewrite the above inequality as
\[ q-1 < (p-1)\big(1+\frac{1}{b}\big) + \frac{c' - c}{a}b. \]

Note that since $\gamma - \gamma' \notin \ZZ$, we have $a \neq a'$. Now by \Cref{prop:order-neg}, from $\gamma \prec_r \gamma'$ it follows that $a<a'$.  By further recalling that $a<a'\leq b$, we have $\frac{a'}{a} = \big(1 + \frac{a'-a}{a}\big) \geq \big(1 + \frac{1}{b}\big)$. Then from the above it follows that
\[ q-1 < (p-1)\frac{a'}{a} + \frac{c' - c}{a}b. \]
Since $\frac{a}{b} > 0$, we can multiply the inequality by $\frac{a}{b}$ to obtain
\[ \frac{(q-1)a}{b} + c  < \frac{(p-1)a'}b + c'. \]
By the assumption that $p, q \in b\mathbb{N}+1$ with $p, q > N_r$, we can now use \Cref{prop:representation} to conclude that $\rem_{q}(\gamma) < \rem_{p}(\gamma')$.
\end{proof}

\subsection{The Main Theorem}

In order to prove our main theorem, we first need a preliminary result on the number of primes in an interval in an arithmetic progression. To this end, we rely on effective bounds on the density of primes in an arithmetic progression.

Given coprime numbers $a,n \in \NN$ with $a<n$, write $\pi_{n,a}(x)$ for the number of primes less than $x$ that are congruent to $a$ modulo $n$. The following estimates on $\pi_{n,a}(x)$ can be found in \cite[Theorem~1.3]{Bennett}. 

\begin{theorem}
\label{th:density}
	Given $n\geq 3$ and $a\in\NN$ coprime to $n$, there exist explicit positive constants $c$ and $x_0$ depending on $n$ such that
	\[ \left| \pi_{n,a}(x) - \frac{Li(x)}{\varphi(n)}\right| < c\frac{x}{(\log x)^2} \text{ for all } x > x_0. \]
\end{theorem}

Note that in the above theorem $Li(x)$ denotes the \emph{offset logarithmic integral function}, which is defined as
\[ Li(x) = \int_{2}^{x}\frac{dt}{\log t}.\]
Asymptotically, the above function behaves as the prime number counting function $\pi(x)$, that is, as $O\big(\frac{x}{\log x}\big)$.
	
\begin{proposition}\label{prop:primes_progression}
		Let $b \in \NN$ and $C\in \ZZ$. There exist an effectively computable bound $M \in \NN$ such that for all primes $p \in b\NN+1$ greater than $M$, there exists a prime $q\in b\NN+1$  with $p< q < p + \frac{p-1}{b}+C$.
\end{proposition}
	
\begin{proof}
Let $x \in b\NN + 1$, and denote by $y = x (1 + \frac{1}{b})$. We would like to show that there exists a bound $M \in \NN$ such that $\pi_{b,1}(y) - \pi_{b,1}(x) > 0$ for all $x > M$. Using \Cref{th:density} to estimate $\pi_{b,1}(y)$ and $\pi_{b,1}(x)$, it suffices to show that:

\[ \frac{Li(y)}{\varphi(b)} - c\frac{y}{(\log y)^2} >  \frac{Li(x)}{\varphi(b)} + c'\frac{x}{(\log x)^2}  
\]
The above then simplifies to:

\begin{equation} \label{eq:asymp}
	\frac{Li(y) - Li(x)}{\varphi(b)} >  2c\frac{y}{(\log y)^2}
\end{equation}
Now write $y = x(1 + \epsilon)$ with $\epsilon = \frac{1}{b}$, and observe that 
\[ Li(y) - Li(x) = \int_{x}^{y} \frac{dt}{\log t} = \frac{y}{\log y} - \frac{x}{\log x} \sim \frac{\epsilon x}{\log x} \]
where $\sim$ denotes asymptotic equivalence.
As for the right hand side of \eqref{eq:asymp}, note that
\[ \frac{y}{(\log y)^2} \sim \frac{x}{(\log x)^2}. \]
It remains to note that $\frac{x}{\log x} \gg \frac{x}{(\log x)^2}$.
\end{proof}

The above proposition ensures that we can construct an infinite sequence of primes $\langle p_i\rangle_{i = 0}^{\infty}$ in $b\NN + 1$ such that for all $p_i$ its successor $p_{i+1}$ is between $p_i$ and $p_i + \frac{p_i-1}{b} + C$.~\Cref{prop:constructqfromp} then implies that if we choose 
$p_0 > \max(M,N_r)$,
for every $\gamma,\gamma' \in \As \uplus\Bs$ such that 
$\overline{\rm (\gamma,\gamma')}$ is  
an $r$-expanding $r$-unbalanced family of intervals,  for all $i$, the intervals $ \{n\in\NN : \rem_{p_i}(\gamma)\leq n <\rem_{p_i}(\gamma')\}$ and $\{ n \in \NN : \rem_{p_{i+1}}(\gamma) \leq n < \rem_{p_{i+1}}(\gamma')\}$ are contiguous. 

The sequence $\langle p_i\rangle_{i = 0}^{\infty}$ is the last part of our construction, allowing us to prove our main result:

\begin{theorem}\label{th:main_result}
	The membership problem for hypergeometric sequences \emph{with rational parameters} is decidable.
\end{theorem}

\begin{proof}
As discussed in Section~\ref{subsec:shiftquo}, the only case of MP that is not trivially decidable is when the shift quotient $r(x)\in \QQ(x)$ of the sequence $\langle u_n \rangle_{n=0}^\infty$ converges to $\pm 1$ as $x\to\infty$.
Given such an instance of MP, we write $r(x)$ as
\[\frac{(x-\alpha_1)\cdots (x-\alpha_d)}{(x-\beta_1)\cdots (x-\beta_d)}.\]
We denote by $\As$ the multiset $\{\alpha_1,\ldots,\alpha_d\}$ consisting of all the (possibly repeated) roots of the numerator and by $\Bs$ the multiset $\{\beta_1,\ldots,\beta_d\}$ of the roots of the denominator.
As discussed in \Cref{sec:mp}, all elements in $\As \uplus \Bs$ are in $\QQ \setminus \ZZ_{\geq 0}$. 
 
We now show that there exists a bound $N\in \NN$ such that for all $n> N$, there exists a prime $p$ appearing in the factorisation of $u_n$ 
but  not in the factorisation of $u_0$ or that of $t$. 
This implies that it suffices to check whether $u_n = t$ for all $n \in \{1,\ldots,N\}$ to decide MP. In particular, we show that for all $n > N$, we can find a prime $p$ with 
$v_p(t) =v_p(u_0) = 0$
such that $ S_p (n)$, as defined in \Cref{eq:summ}, is non-zero.

Write $u_0 = \frac{v}{v'}$ and $t = \frac{w}{w'}$. Now define
\[ N' = \max(|v|,|v'|,|w|, |w'|,N_r)\]
where $N_r$ is defined as in \Cref{eq:nr}. Note that
 none of the primes $p > N'$ divide the target~$t$ nor the initial value $u_0$.
	
Following \Cref{prop:reduction}, we can assume without loss of generality that there exists an $r$-expanding $r$-unbalanced family of intervals $\overline{\rm (\gamma,\gamma')}$ for some $\gamma,\gamma'\in\As \uplus \Bs$. Let $M$ be the bound computed in \Cref{prop:primes_progression}. Let $p_0 \in b\NN+1$ be a prime with $p_0 > \max(M,N')$.
 Observe that for all $n$ in
\[ \{k \in \NN : \rem_{p_0}(\gamma) \leq k < \rem_{p_0}(\gamma')\},\]
 the sum $S_{p_0}(n)$ defined in \Cref{eq:summ} is indeed non-zero. 

Following \Cref{prop:primes_progression}, we can construct an infinite sequence of primes $\langle p_i\rangle_{i = 0}^{\infty}$ with initial element $p_0$ such that for every prime $p_i$ in the sequence, its successor $p_{i+1}$ is between $p_i$ and $p_i + \frac{p_i-1}{b}$ in the arithmetic progression $b\NN + 1$. By \Cref{prop:constructqfromp}, for all $i$, the intervals $ \{k \in \NN : \rem_{p_i}(\gamma) \leq k < \rem_{p_i}(\gamma')\}$ and $ \{k \in \NN :\rem_{p_{i+1}}(\gamma) \leq k < \rem_{p_{i+1}}(\gamma')\}$ are contiguous.
Therefore we can cover all $n>\rem_{p_0}(\gamma)$, which concludes our proof.
	
\end{proof}

\section{Discussion} \label{sec:discussion}

Our main result shows decidability of the Membership Problem for hypergeometric sequences, where the coefficient polynomials have rational roots. In extending our results 
to the general case, that is when the coefficient polynomials have algebraic roots, it appears that the difficulty is when the splitting fields of the coefficient polynomials are identical. Otherwise, there exists a prime~$p$ such that one of 
polynomials splits over $\QQ_p$ but the other does not.  Such a prime can be used to deduce an  upper bound on the largest index of the target value.

\SkipTocEntry\section*{Acknowledgements}
  A. Pouly is supported by the European Research Council (ERC) under the 
  European Union's Horizon 2020 research and innovation programme (Grant 
  agreement No. 852769, ARiAT).

\bibliographystyle{plain}
\bibliography{literature}

\begin{thebibliography}{10}

\bibitem{AMM07}
T.~Amdeberhan, L.~Medina, and V.~Moll.
\newblock Asymptotic valuations of sequences satisfying first order
  recurrences.
\newblock {\em Proceedings of the American Mathematical Society}, 137:885--890,
  10 2007.

\bibitem{Andre2004}
Y.~Andr{\'e}.
\newblock {\em Une introduction aux motifs (motifs purs, motifs mixtes,
  p{\'e}riodes)}, volume~17 of {\em Panoramas et Synth{\`e}ses}.
\newblock Soci{\'e}t{\'e} Math{\'e}matique de France, 2004.

\bibitem{BBY11}
J.~Bell, S.~Burris, and K.~Yeats.
\newblock On the set of zero coefficients of a function satisfying a linear
  differential equation.
\newblock {\em Mathematical Proceedings of the Cambridge Philosophical
  Society}, 153:235--247, 05 2011.

\bibitem{Bennett}
M.~A. Bennett, G.~Martin, K.~{O'Bryant}, and A.~Rechnitzer.
\newblock {Explicit bounds for primes in arithmetic progressions}.
\newblock {\em Illinois Journal of Mathematics}, 62(1-4):427 -- 532, 2018.

\bibitem{ChamberlandStraub13}
M.~Chamberland and A.~Straub.
\newblock On gamma quotients and infinite products.
\newblock {\em Advances in Applied Mathematics}, 51(5):546--562, 2013.

\bibitem{DRR13}
E.~Delaygue, T.~Rivoal, and J.~Roques.
\newblock On dwork’s p-adic formal congruences theorem and hypergeometric
  mirror maps.
\newblock {\em Memoirs of the American Mathematical Society}, 246, 09 2013.

\bibitem{KauersP11}
M.~Kauers and P.~Paule.
\newblock {\em The Concrete Tetrahedron - Symbolic Sums, Recurrence Equations,
  Generating Functions, Asymptotic Estimates}.
\newblock Texts {\&} Monographs in Symbolic Computation. Springer, 2011.

\bibitem{KauersP10}
M.~Kauers and V.~Pillwein.
\newblock When can we detect that a p-finite sequence is positive?
\newblock In Wolfram Koepf, editor, {\em Symbolic and Algebraic Computation,
  International Symposium, {ISSAC}, Proceedings}, pages 195--201. {ACM}, 2010.

\bibitem{kenison2020positivity}
G.~Kenison, O.~Klurman, E.~Lefaucheux, F.~Luca, P.~Moree, J.~Ouaknine,
  A.~Whiteland, and J.~Worrell.
\newblock On inequality decision prolems for low-order holonomic sequences,
  2020.
\newblock Submitted.

\bibitem{KenisonKLLMOW021}
G.~Kenison, O.~Klurman, E.~Lefaucheux, F.~Luca, P.~Moree, J.~Ouaknine, M.~A.
  Whiteland, and J.~Worrell.
\newblock On positivity and minimality for second-order holonomic sequences.
\newblock In {\em 46th International Symposium on Mathematical Foundations of
  Computer Science, {MFCS}}, volume 202 of {\em LIPIcs}, pages 67:1--67:15.
  Schloss Dagstuhl - Leibniz-Zentrum f{\"{u}}r Informatik, 2021.

\bibitem{KoblitzOgus1979}
N.~Koblitz and A.~Ogus.
\newblock Algebraicity of some products of values of the f function, appendix
  of valeurs de fonctions l et periodes d'int\'egrales.
\newblock {\em Proceedings of Symposia in Pure Mathematics}, 33:part 2,
  313--346, 1979.

\bibitem{Lang1977}
S.~Lang.
\newblock Relations de distributions et exemples classiques.
\newblock {\em S\'eminaire Delange-Pisot-Poitou. Th\'eorie des nombres}, 19(2),
  1977-1978.

\bibitem{MST84}
M.~Mignotte, T.~N. Shorey, and R.~Tijdeman.
\newblock The distance between terms of an algebraic recurrence sequence.
\newblock {\em Journal f\"ur die reine und angewandte Mathematik}, 349, 1984.

\bibitem{Nagura}
J.~Nagura.
\newblock On the interval containing at least one prime number.
\newblock {\em Proceedings of the Japan Academy}, 28(4):177--181, 1952.

\bibitem{NeumannO021}
E.~Neumann, J.~Ouaknine, and J.~Worrell.
\newblock Decision problems for second-order holonomic recurrences.
\newblock In {\em 48th International Colloquium on Automata, Languages, and
  Programming, {ICALP}}, volume 198 of {\em LIPIcs}, pages 99:1--99:20. Schloss
  Dagstuhl - Leibniz-Zentrum f{\"{u}}r Informatik, 2021.

\bibitem{PetkovsekWilfZeilberger1996}
M.~Petkovsek, H.~Wilf, and D.~Zeilberger.
\newblock {\em {$A=B$}}.
\newblock A K Peters, 1997.

\bibitem{PillweinS15}
V.~Pillwein and M.~Schussler.
\newblock An efficient procedure deciding positivity for a class of holonomic
  functions.
\newblock {\em {ACM} Commun. Comput. Algebra}, 49(3):90--93, 2015.

\bibitem{Ver85}
N.~K. Vereshchagin.
\newblock The problem of appearance of a zero in a linear recurrence sequence
  (in {R}ussian).
\newblock {\em Mat. Zametki}, 38(2), 1985.

\bibitem{Waldschmidt06}
M.~Waldschmidt.
\newblock {Transcendence of periods: the state of the art}.
\newblock {\em {Pure and Applied Mathematics Quarterly}}, 2(2):435--463, 2006.

\end{thebibliography}

\appendix

\section{Membership of zero and assumption on~$q(x)$}
\label{appendix:zero}

In our MP instances, we will assume that the numerator~$q(x)$ of the shift quotient
 has no non-negative integer zeros and $t\neq 0$.
This  assumption on $q(x)$ is without loss of generality as otherwise, the sequence $\langle u_n \rangle_{n=0}^{\infty}$ will be ultimately always zero.
Indeed, if $q(x)$ has non-negative integer zeros, $u_n = 0$ for all $n \geq m$ where $m$ is the smallest  non-negative integer root of~$q(x)$.
Consequently, the search domain for indices  of~$t$ in MP will be limited to the finite set $\{u_0,\ldots, u_m\}$.
This assumption on~$q(x)$ will exclude the membership of zero in  the sequence~$\langle u_n \rangle_{n=0}^{\infty}$ , and will allow us to further assume that $t\neq 0$.

\section{Proof of Proposition~\ref{prop:reduction}}
\label{appendix:reduction}


Recall we can write the shift quotient $r(x)$ as
\[
	\frac{(x-\alpha_1)\cdots (x-\alpha_d)}{(x-\beta_1)\cdots (x-\beta_d)}
\]
where the $\alpha_i$ and the $\beta_i$ are in  $\QQ \setminus \ZZ_{\geq 0}$. 
We denote by $\As$ the multiset $\{\alpha_1,\ldots,\alpha_d\}$ consisting of all the (possibly repeated) roots of the numerator and by $\Bs$ the multiset $\{\beta_1,\ldots,\beta_d\}$ of the roots of the denominator.

Our proof relies on the following observation: if there exists a bijective function $f : \As \to \Bs$ such that for all $\alpha\in\As$ we have $\alpha - f(\alpha) \in \ZZ$, then Item~2 holds. Otherwise, Item~1 holds.

Towards Item~2, assume that there exists a bijective function $f : \As \to \Bs$ such that for all $\alpha\in\As$ we have $\alpha - f(\alpha) \in \ZZ$. Given an enumeration $\alpha_1,\ldots,\alpha_d$ on $\As$, enumerate the elements of $\Bs$ so that $f(\alpha_i)=\beta_i$. Then given a pair $\alpha_i,\beta_i$, write $\ell_i=|\alpha_i - \beta_i|\in\NN$. 

Recall that given an instance of MP with a target value $t\in\QQ$ and initial value $u_0\in\QQ$, the problem asks to decide whether there exists $n\in\NN$ such that
\begin{equation}\label{eq:mp_reduction}
	u_0 \prod_{k=1}^{n} r(k)  = t.
\end{equation}
Observe that we can expand a part of the above product in the following way:
\begin{align*}
	\prod_{k=1}^{n} r(k)
	&= \prod_{k=1}^{n} \frac{(k-\alpha_1)\cdots(k-\alpha_d)}{(k-\beta_1)\cdots(k-\beta_d)} \\
	&= \prod_{k=1}^{n} \frac{(k-\alpha_1)}{(k-\beta_1)} \cdots \prod_{k=1}^{n} \frac{(k-\alpha_d)}{(k-\beta_d)}
\end{align*}

Now if we look at the product $\prod_{k=1}^{n} \frac{(k-\alpha_i)}{(k-\beta_i)}$. If $\beta_i>\alpha_i$, observe that we can write 
\[
\frac{k-\alpha_i}{k-\beta_i}=\frac{k+ (\beta_i - \alpha_i)-\beta_i}{k-\beta_i}= \frac{k + \ell_i - \beta_i}{k- \beta_i}.
\]

We can thus write:
\begin{align*}
	&\prod_{k=1}^{n} \frac{(k-\alpha_i)}{(k-\beta_i)}
	\\
	&=\prod_{k=1}^{n} \frac{(k+\ell_i-\beta_i)}{(k-\beta_i)} 
	\\
	&= \frac{(1+\ell_i-\beta_i)}{(1-\beta_i)} \cdots 
	\frac{((\ell_i+1)+\ell_i-\beta_i)}{((\ell_i+1)-\beta_i)} \cdots
	\frac{(n+\ell_i-\beta_i)}{(n-\beta_i)}
	\\
	&= \frac{\cancel{(1+\ell_i-\beta_i)}}{1-\beta_i)} \cdots 
	\frac{((\ell_i+1)+\ell_i-\beta_i)}{\cancel{((\ell_i+1)-\beta_i)}} \cdots
	\frac{(n+\ell_i-\beta_i)}{(n-\beta_i)}
\end{align*}

Note how in the last line above, we were able to simplify at least two terms. Observe that we will be able to do so for all but $\ell_i$ terms in the numerator and all but $\ell_i$ terms in the denominator. That is, we transform:
\begin{align*}
	\prod_{k=1}^{n} \frac{(k-\alpha_i)}{(k-\beta_i)}
	&= \frac{(n+1-\beta_i)\cdots (n+\ell_i-\beta_i)}{(1-\beta_i)\cdots (\ell_i-\beta_i)} 
	= \frac{f_i(n)}{g_i(n)}
\end{align*}
for $n>\ell_i$.

Similarly, observe that if $\alpha_i>\beta_i$, we can write 
\[
\frac{k-\alpha_i}{k-\beta_i}=\frac{k-\alpha_i}{k+(\alpha_i - \beta_i)-\alpha_i}= \frac{k - \alpha_i}{k+\ell_i-\alpha_i}
\]
By repeating the above computation, we obtain:
\begin{align*}
	\prod_{k=1}^{n} \frac{(k-\alpha_i)}{(k-\beta_i)} &= \frac{(1-\alpha_i)\cdots (\ell_i-\alpha_i)}{(n+1-\alpha_i)\cdots (n+\ell_i-\alpha_i)}
	= \frac{f_i(n)}{g_i(n)}
\end{align*}
for $n>\ell_i$.

By applying the above transformation to all pairs $\alpha_i,\beta_i$ with their respective distance $\ell_i$, we can rewrite $\prod_{k=1}^{n} r(k)$ as:
\begin{align*}
	\prod_{k=1}^{n}r(k) = \frac{f_1(n)}{g_1(n)} \ldots \frac{f_d(n)}{g_d(n)} = \frac{\hat{f}(n)}{\hat{g}(n)}
\end{align*}
The above is well-defined for all $n>\max(\ell_1,\ldots,\ell_d)$, and implies that the product $\prod_{k=1}^{n}r(k)$ can be decomposed as a product of a constant number of rational functions in $n$.

To show Item~1, assume that there is no  bijective function $f : \As \to \Bs$ such that for all $\alpha\in\As$ we have $\alpha - f(\alpha) \in \ZZ$.
Let $\gamma_{1},\ldots,$ $\gamma_{2d}$ be a fixed permutation of the elements of  $\As \uplus \Bs$, such that $\gamma_{j} \preceq_r \gamma_{k}$ for all $1 \leq j < k \leq 2d$.

We claim that there exists an index $j$ such that $\overline{\rm (\gamma_j,\gamma_{j+1})}$ is an $r$-expanding $r$-unbalanced family of intervals.
 That is, $j$ is such that
\begin{itemize}
	\item $\gamma_j - \gamma_{j+1}\notin \ZZ$, and
	\item the number of $\alpha_i$ in the block $\gamma_1 \preceq_r \ldots \preceq_r \gamma_{j}$ is not equal to the number  of $\beta_i$ in the block $\gamma_1 \preceq_r \ldots \preceq_r \gamma_{j}$.
\end{itemize}
Indeed, if there is no such $j$, then we can take each block of the $\alpha_i$ and $\beta_i$ with integer distances and construct a bijection mapping from the set of $\alpha_i$'s to the set of $\beta_i$'s appearing in the block alone. Putting together the mappings given by the block bijections gives us a bijection $f : \As \to \Bs$ such that $\alpha_i - f(\alpha_i) \in \ZZ$ for all $i$, which would lead to a contradiction.

%
%
%

\end{document}